\documentclass[graybox]{svmult}

\usepackage{mathptmx}       % selects Times Roman as basic font
\usepackage{helvet}         % selects Helvetica as sans-serif font
\usepackage{courier}        % selects Courier as typewriter font
\usepackage{type1cm}        % activate if the above 3 fonts are
\usepackage{makeidx}         % allows index generation
\usepackage{graphicx}        % standard LaTeX graphics tool
\usepackage{multicol}        % used for the two-column index
\usepackage[bottom]{footmisc}% places footnotes at page bottom

\usepackage{paralist}

\usepackage{empheq}
\usepackage{amstext, amsmath, amssymb}

\usepackage[%
    bookmarks=true,         % show bookmarks bar?
    unicode=true,          % non-Latin characters in Acrobat’s bookmarks
    pdftoolbar=true,        % show Acrobat’s toolbar?
    pdfmenubar=true,        % show Acrobat’s menu?
    pdffitwindow=false,     % window fit to page when opened
    pdfstartview={FitH},    % fits the width of the page to the window
    pdftitle={An Exact Formula for the Average Run Length to False Alarm of the Generalized Shiryaev--Roberts Procedure for Change-Point Detection under Exponential Observations},    % title
    pdfauthor={Wenyu Du and Grigory Sokolov and Aleksey S. Polunchenko},     % author
    pdfsubject={An Exact Formula for the Average Run Length to False Alarm of the Generalized Shiryaev--Roberts Procedure for Change-Point Detection under Exponential Observations},   % subject of the document
    pdfcreator={Aleksey S. Polunchenko},   % creator of the document
    pdfproducer={MikTeX}, % producer of the document
    pdfkeywords={Sequential Analysis}{Quickest change-point detection}{Shiryaev--Roberts procedure}{Generalized Shiryaev--Roberts procedure}{Sequential analysis}{Quality control}{Integral equations}{Average run length to false alarm}, % list of keywords
    pdfnewwindow=true,      % links in new window
    colorlinks=false,       % false: boxed links; true: colored links
    linkcolor=red,          % color of internal links (change box color with linkbordercolor)
    citecolor=green,        % color of links to bibliography
    filecolor=magenta,      % color of file links
    urlcolor=cyan           % color of external links
]{hyperref}

\usepackage[T1]{fontenc}

\renewcommand{\Pr}{\mathbb{P}} % probability
\DeclareMathOperator{\EV}{\mathbb{E}} % expected value

\DeclareMathOperator{\LR}{\Lambda}

\DeclareMathOperator{\ARL}{ARL}

\DeclareMathOperator{\One}{\mathchoice{\rm 1\mskip-4.2mu l}{\rm 1\mskip-4.2mu l}{\rm 1\mskip-4.6mu l}{\rm 1\mskip-5.2mu l}}
\newcommand{\indicator}[1]{{\One_{\left\{#1\right\}}}}

\begin{document}

\title*{An Exact Formula\texorpdfstring{\\}{}
for the Average Run Length to False Alarm\texorpdfstring{\\}{}
of the Generalized Shiryaev--Roberts Procedure\texorpdfstring{\\}{}
for Change-Point Detection\texorpdfstring{\\}{}
under Exponential Observations}
\titlerunning{On the ARL to False Alarm of the Generalized Shiryaev--Roberts Procedure}
\author{Wenyu Du and Grigory Sokolov and Aleksey S. Polunchenko}
\authorrunning{Du, Sokolov and Polunchenko}
\institute{W. Du, G. Sokolov and A.S. Polunchenko \at Department of Mathematical Sciences, State University of New York at Binghamton, Binghamton, New York 13902--6000, USA e-mail:\texttt{\{wdu1,gsokolov,aleksey\}@binghamton.edu}}
%
%\institute{Wenyu Du \at Department of Mathematical Sciences, State University of New York at Binghamton, Binghamton, New York, USA \email{wdu1@binghamton.edu}
%\and Grigory Sokolov \at Department of Mathematical Sciences, State University of New York at Binghamton, Binghamton, New York, USA \email{gsokolov@binghamton.edu}
%\and Aleksey S. Polunchenko \at Department of Mathematical Sciences, State University of New York at Binghamton, Binghamton, New York, USA \email{aleksey@binghamton.edu}}
\maketitle
\abstract{We derive analytically an exact closed-form formula for the standard minimax Average Run Length (ARL) to false alarm delivered by the Generalized Shiryaev--Roberts (GSR) change-point detection procedure devised to detect a shift in the baseline mean of a sequence of independent exponentially distributed observations. Specifically, the formula is found through direct solution of the respective integral (renewal) equation, and is a general result in that the GSR procedure's headstart is not restricted to a bounded range, nor is there a ``ceiling'' value for the detection threshold. Apart from the theoretical significance (in change-point detection, exact closed-form performance formulae are typically either difficult or impossible to get, especially for the GSR procedure), the obtained formula is also useful to a practitioner: in cases of practical interest, the formula is a function linear in both the detection threshold and the headstart, and, therefore, the ARL to false alarm of the GSR procedure can be easily computed.
}
% starred version:
\abstract*{We derive analytically an exact closed-form formula for the standard minimax Average Run Length (ARL) to false alarm delivered by the Generalized Shiryaev--Roberts (GSR) change-point detection procedure devised to detect a shift in the baseline mean of a sequence of independent exponentially distributed observations. Specifically, the formula is found through direct solution of the respective integral (renewal) equation, and is a general result in that the GSR procedure's headstart is not restricted to a bounded range, nor is there a ``ceiling'' value for the detection threshold. Apart from the theoretical significance (in change-point detection, exact closed-form performance formulae are typically very difficult or impossible altogether to get, especially for the GSR procedure), the obtained formula is also useful to a practitioner: in cases of practical interest, the formula is a function linear in both the detection threshold and the headstart, and, therefore, the ARL to false alarm of the GSR procedure can be easily computed.
}
%Each chapter should be preceded by an abstract (10--15 lines long) that summarizes the content. The abstract will appear \textit{online} at \url{www.SpringerLink.com} and be available with unrestricted access. This allows unregistered users to read the abstract as a teaser for the complete chapter. As a general rule the abstracts will not appear in the printed version of your book unless it is the style of your particular book or that of the series to which your book belongs.\newline\indent
%Please use the 'starred' version of the new Springer \texttt{abstract} command
%for typesetting the text of the online abstracts and include them with the source files of your manuscript. Use the plain \texttt{abstract} command if the abstract is also to appear in the printed version of the book

%+-----------------------------------------------------------------------------------------------+%
\section{Introduction}
\label{sec:intro}
Quickest change-point detection is concerned with the design and analysis of reliable statistical machinery for rapid detection of changes that may spontaneously affect a ``live'' process, continuously monitored via sequentially made observations. See, e.g.,~\cite{Poor+Hadjiliadis:Book2009} or~\cite[Part~II]{Tartakovsky+etal:Book2014}. A quickest change-point detection procedure is a stopping time adapted to the observed data, and is a rule whereby one is to stop and ``sound an alarm'' that the characteristics of the observed process may have (been) changed. A ``good'' (i.e., optimal or nearly optimal) detection procedure is one that minimizes (or nearly minimizes) the desired detection delay penalty, subject to a constraint on the false alarm risk. For an overview of the major optimality criteria see, e.g.,~\cite{Tartakovsky+Moustakides:SA2010,Polunchenko+Tartakovsky:MCAP2012,Polunchenko+etal:JSM2013,Veeravalli+Banerjee:BookCh2013} or~\cite[Part~II]{Tartakovsky+etal:Book2014}.

A problem particularly persistent in applied change-point detection (e.g., in quality control) is evaluation of detection procedures' performance. To that end, the ideal would be to have the needed performance metrics expressed exactly and in a closed and simple form. However, this is generally quite difficult mathematically, if at all possible. Part of the reason is that the renewal equations that many popular performance metrics satisfy are Fredholm integral equations of the second kind (possibly written as equivalent differential equations), and such equations seldom allow for an analytical solution. As a result, the standard practice has been to evaluate the performance numerically (one particularly popular approach has been to devise an asymptotic approximation of some sort). Nevertheless, some exact performance formulae have been derived explicitly, although primarily for the ``mainstream'' detection methods. For instance, a number of characteristics of the celebrated CUSUM ``inspection scheme'' (due to~\cite{Page:B1954}) have been expressed explicitly, e.g., in~\cite{Regula:PhD-Thesis1975,Vardeman+Ray:T1985,Gan:SS1992,Knoth:Thesis1995,DeLucia+Poor:IEEE-IT1997,Knoth:SA1998}\footnote{\label{fnt:literature}By no means is this an exhaustive list of available papers on the subject.}, although for only a handful of scenarios. Likewise, exact closed-form formulae for various performance metrics of the famous EWMA chart (due to~\cite{Roberts:T1959}) in an exponential scenario have been established, e.g., in~\cite{Novikov:TPA1990,Gan:JQT1998,Polunchenko+etal:SLJAS2014}\footref{fnt:literature}.% The local structure of the Shewhart chart makes it possible to solve the corresponding integral equations analytically.

However, the corresponding progress made to date for the classical Shiryaev--Roberts (SR) procedure (due to~\cite{Shiryaev:SMD1961,Shiryaev:TPA1963,Roberts:T1966}) is far more modest (except for the continuous-time case), and especially little has been done for the Generalized SR (GSR) procedure, which was introduced recently in~\cite{Moustakides+etal:SS2011} as a ``headstarted'' version of the classical SR procedure. Since the latter is a special case of the GSR procedure (when the headstart is zero), from now on we will follow~\cite{Tartakovsky+etal:TPA2012} and use the term ``GSR procedure'' to refer to both procedures. As a matter of fact, to the best of our knowledge, exact and explicit formulae for a small subset of characteristics of the GSR procedure have been obtained only in~\cite{Pollak:AS1985,Kenett+Pollak:IEEE-TR1986,Mevorach+Pollak:AJMMS1991,Yakir:AS1997,Mei:AS2006,Tartakovsky+Polunchenko:IWAP2010,Polunchenko+Tartakovsky:AS2010,Polunchenko+Tartakovsky:MCAP2012}. The purpose of this work is to add on to this list. Specifically, we obtain an exact, closed-form formula for the standard (minimax) Average Run Length (ARL) to false alarm delivered by the GSR procedure devised to detect a jump in the common baseline mean of a sequence of independent exponentially distributed observations. The formula is found analytically, through direct solution of the respective renewal (integral) equation, and is valid for an arbitrary (nonnegative) headstart, with the detection threshold not restricted from above. Furthermore, the formula is remarkably simple (it is a function linear in the detection threshold and in the headstart) and, unlike its complicated and cumbersome CUSUM and EWMA counterparts, {\em can} be used to compute the GSR procedure's ARL to false alarm (in the exponential scenario) ``{\em by hand}''. This would clearly be of aid to a practitioner.

%+-----------------------------------------------------------------------------------------------+%
\section{Preliminaries}
\label{sec:preliminaries}
The centerpiece of this work is the (minimax) Average Run Length (ARL) to false alarm of the Generalized Shiryaev--Roberts (GSR) detection procedure (due to~\cite{Moustakides+etal:SS2011}) considered in the context of the basic minimax quickest change-point detection problem (see, e.g.,~\cite{Lorden:AMS1971,Pollak:AS1985}). As a performance metric, the ARL to false alarm was apparently introduced in~\cite{Page:B1954}; see also, e.g.,~\cite{Lorden:AMS1971}.

Let $f_\infty(x)$ and $f_0(x)$ denote, respectively, the observations' pdf in the pre- and post-change regime. Let $\LR_n\triangleq f_0(X_n)/f_\infty(X_n)$ be the ``instantaneous'' likelihood ratio (LR) for the $n$-th data point, $X_n$. The GSR procedure (due to~\cite{Moustakides+etal:SS2011}) is then formally defined as the stopping time
\begin{equation}\label{eq:T-SR-r-def}
\mathcal{S}_{A}^r
\triangleq
\inf\{n\ge1\colon R_n^r \ge A\},\;\text{such that}\;\inf\{\varnothing\}=\infty,
\end{equation}
where $A>0$ is a detection threshold used to control the false alarm risk, and
\begin{equation}\label{eq:Rn-SR-r-def}
R_{n+1}^r
=
(1+R_{n}^r)\LR_{n+1}\;\text{for}\; n=0,1,\ldots\;\text{with}\; R_0^r=r\ge0,
\end{equation}
is the GSR detection statistic. We remark that $R_0^r=r\ge0$ is a design parameter referred to as the headstart and, in particular, when $R_0^r=r=0$, the GSR procedure is equivalent to the classical Shiryaev--Roberts (SR) procedure (due to~\cite{Shiryaev:SMD1961,Shiryaev:TPA1963,Roberts:T1966}); a brief account of history of the SR procedure may be found, e.g., in~\cite{Pollak:IWSM2009}. Albeit ``young'' (the GSR procedure was proposed in 2011), it has already been shown (see, e.g.,~\cite{Pollak+Tartakovsky:SS2009,Shiryaev+Zryumov:Khabanov2010,Tartakovsky+Polunchenko:IWAP2010,Polunchenko+Tartakovsky:AS2010,Tartakovsky+etal:TPA2012}) to possess very strong optimality properties, not exhibited by the CUSUM scheme or the EWMA chart; in fact, in certain scenarios, the latter two charts have been found experimentally to be inferior to the GSR procedure.

Let $\Pr_\infty$ ($\EV_\infty$) be the probability measure (expectation) induced by the observations in the pre-change regime, i.e., when $X_n\propto f_\infty(x)$ for {\em all} $n\ge1$.  The ARL to false alarm of the GSR procedure is defined as $\ARL(\mathcal{S}_A^r)\triangleq\EV_\infty[\mathcal{S}_A^r]$. A key property of the GSR statistic~\eqref{eq:Rn-SR-r-def} is that the sequence $\{R_n^r-n-r\}_{n\ge0}$ is a zero-mean $\Pr_\infty$-martingale, i.e., $\EV_\infty[R_n^r-n-r]=0$ for all $n\ge0$ and all $r$. This and Doob's Optional stopping (sampling) theorem (see, e.g.,~\cite[Theorem~2.3.1,~p.~31]{Tartakovsky+etal:Book2014}) imply that $\EV_\infty[R_{\mathcal{S}_A^{r}}-\mathcal{S}_A^{r}-r]=0$, so that $\ARL(\mathcal{S}_A^{r})=\EV_\infty[R_{\mathcal{S}_A^{r}}]-r\ge A-r$. As a result, to ensure that $\ARL(\mathcal{S}_A^{r})\ge\gamma$ for a desired $\gamma>1$, it suffices to pick $A$ and $r$ from the solution set of the inequality  $A-r\ge\gamma$ and such that $A>0$ and $r\ge0$.

A more accurate result is the approximation $\ARL(\mathcal{S}_A^r)\approx A/\xi-r$ valid for sufficiently large $A>0$; see, e.g.,~\cite[Theorem~1]{Pollak:AS1987} or~\cite{Tartakovsky+etal:TPA2012}. To define $\xi$, let $S_n\triangleq\sum_{i=1}^n\log\LR_n$ for $n\ge1$, and let $\tau_a\triangleq\inf\{n\ge1\colon S_n\ge a\}$ for $a>0$ (again, with the understanding that $\inf\{\varnothing\}=\infty$). Then $\kappa_a\triangleq S_{\tau_a}-a$ is the so-called ``overshoot'' (excess over the level $a>0$ at stopping), and $\xi\triangleq\lim_{a\to\infty}\EV_0[e^{-\kappa_a}]$, and is referred to as the ``limiting average exponential overshoot''; here $\EV_0$ denotes the expectation under the probability measure induced by the observations in the post-change regime, i.e., when $X_n\propto f_0(x)$ for {\em all} $n\ge1$. In general, $\xi$ is clearly between $0$ and $1$, and is a model-dependent constant, which falls within the scope of nonlinear renewal theory; see, e.g.,~\cite{Woodroofe:Book1982},~\cite[Section~II.C]{Veeravalli+Banerjee:BookCh2013} or~\cite[Section~2.6]{Tartakovsky+etal:Book2014}.

We now state the main equation that we shall deal with (and, in fact, solve analytically) in the next section in a certain exponential scenario. Let $P_\infty^{\LR}(t)\triangleq\Pr_\infty(\LR_1\le t)$, $t\ge0$, be the cdf of the LR under probability measure $\Pr_\infty$. Let $R_0^{r=x}=x\ge0$ be fixed and define
\begin{equation}\label{eq:def-Kinf}
\mathcal{K}_\infty(x,y)
\triangleq
\frac{\partial}{\partial y}
\Pr_\infty(R_{n+1}^r\le y|R_n^r=x)
=
\frac{\partial}{\partial y}P_{\infty}^{\LR}\left(\frac{y}{1+x}\right),\;\text{for}\; x,y\ge0,
\end{equation}
i.e., the transition probability density kernel for the homogeneous Markov process $\{R_n^r\}_{n\ge0}$ under probability measure $\Pr_\infty$.

From now on, let $\ell(x,A)\triangleq\ARL(\mathcal{S}_A^{r=x})$. It is shown, e.g., in~\cite{Moustakides+etal:SS2011}, that $\ell(x,A)$ is governed by the renewal equation
\begin{equation}\label{eq:ARL-int-eqn}
\ell(x,A)
=
1+\int_0^A\mathcal{K}_{\infty}(x,y)\,\ell(y,A)\,dy,
\end{equation}
where $x\ge0$ and $A>0$. The question of existence and uniqueness of solution for this equation has been answered in the affirmative, e.g., in~\cite{Moustakides+etal:SS2011}. It is this equation, viz. the exact solution thereof in a specific exponential scenario, that is the centerpiece of this work.

Equation~\eqref{eq:ARL-int-eqn} is a Fredholm (linear) integral equation of the second kind. Since for such equations an analytical solution is rarely a possibility, they are usually solved numerically. Numerical schemes specifically for equation~\eqref{eq:ARL-int-eqn} have been developed and applied, e.g., in~\cite{Tartakovsky+etal:IWSM2009,Moustakides+etal:SS2011,Polunchenko+etal:SA2014}. However, it turns out that in a certain exponential scenario it is possible to solve~\eqref{eq:ARL-int-eqn} analytically, and, more importantly, the solution is a simple linear function of $x$ and $A$, just as one would expect from the approximation $\ARL(\mathcal{S}_A^{r})\approx A/\xi-r$ mentioned earlier. This is the main result of this paper, it generalizes~\cite[Proposition~1]{Kenett+Pollak:JAP1996}, and the details are given in the next section.

%+-----------------------------------------------------------------------------------------------+%
\section{The Main Result}
\label{sec:main-result}
% Always give a unique label
% and use \ref{<label>} for cross-references
% and \cite{<label>} for bibliographic references
% use \sectionmark{}
% to alter or adjust the section heading in the running head
We are now in a position to establish the main result of this work, i.e., derive analytically an exact closed-form formula for the ARL to false alarm exhibited by the GSR procedure~\eqref{eq:T-SR-r-def}--\eqref{eq:Rn-SR-r-def} ``tasked'' to detect a change in the baseline (common) mean of a series of independent exponentially distributed observations. More concretely, suppose the observations' pre- and post-change pdf's are
\begin{equation}\label{eq:exp-model-def}
f_\infty(x)
=
e^{-x}\indicator{x\ge0}\;\text{and}\;
f_0(x)
=\frac{1}{1+\theta}e^{-x/(1+\theta)}\indicator{x\ge0},
\end{equation}
respectively, where $\theta>0$, a known parameter with an obvious interpretation: it is the magnitude of the shift in the mean of the exponential distribution, so that the higher (lower) the value of $\theta$, the more (less) contrast the mean shift is, and the easier (harder) it is to detect. We shall from now on refer to this scenario as the $\mathcal{E}(1)$-to-$\mathcal{E}(1+\theta)$ model, to reflect not only the throughout ``exponentiality'' of the data, but also that their mean is $1$ pre-change and $1+\theta>1$ post-change. For a motivation to consider this model, see, e.g.,~\cite{Kenett+Pollak:IEEE-TR1986},~\cite{Tartakovsky+Ivanova:PIT1992}, or~\cite[Section~3.1.6]{Tartakovsky+etal:Book2014}.

To ``tailor'' the general equation~\eqref{eq:ARL-int-eqn} on the ARL to false alarm to the $\mathcal{E}(1)$-to-$\mathcal{E}(1+\theta)$ model, the first step is to find $\LR_n\triangleq f_0(X_n)/f_\infty(X_n)$. To that end, it is easy to see from~\eqref{eq:exp-model-def} that
\begin{equation}\label{eq:exp-case-LR-def}
\LR_n
=
\frac{1}{1+\theta}\exp\left\{\frac{\theta}{1+\theta} X_n\right\},\,n\ge1,
\end{equation}
and we note that since $X_n\ge0$ w.p.~1 for all $n\ge1$ under any probability measure, it can be deduced that $\LR_n\ge1/(1+\theta)$ w.p.~1 for all $n\ge1$, also under any probability measure. The latter inequality is a circumstance with consequences, which are illustrated in the following two results.
\begin{lemma}%\label{lem:Kinf-exp-case-formula}
For the $\mathcal{E}(1)$-to-$\mathcal{E}(1+\theta)$ model~\eqref{eq:exp-model-def}, the pre-change transition probability density kernel, $\mathcal{K}_\infty(x,y)$, defined by~\eqref{eq:def-Kinf}, is given by the formula:
\begin{equation}\label{eq:Kinf-exp-case-formula}
\mathcal{K}_\infty(x,y)
=
\theta^{-1}(1+\theta)^{-1/\theta} y^{-2-1/\theta}(1+x)^{1+1/\theta}\indicator{y\ge(1+x)/(1+\theta)},
\end{equation}
where it is understood that $x\ge0$.
\end{lemma}
\begin{proof}
\smartqed
The desired result can be established directly from~\eqref{eq:def-Kinf}, i.e., the definition of the pre-change transition probability density kernel, $\mathcal{K}_\infty(x,y)$, combined with~\eqref{eq:exp-case-LR-def}, i.e., the formula for the LR specific to the $\mathcal{E}(1)$-to-$\mathcal{E}(1+\theta)$ model~\eqref{eq:exp-model-def}. The presence of the indicator function in the right-hand side of~\eqref{eq:Kinf-exp-case-formula} is an implication of the aforementioned inequality $\LR_n\ge 1/(1+\theta)$ valid w.p.~1 for all $n\ge1$ and under any probability measure.
\qed
\end{proof}
Now, with~\eqref{eq:Kinf-exp-case-formula} put in place of $\mathcal{K}_\infty(x,y)$ in the general equation~\eqref{eq:ARL-int-eqn} the latter takes on the form
\begin{equation}\label{eq:exp-case-ARL-int-eqn}
\ell(x,A)
=
1+\theta^{-1}(1+\theta)^{-1/\theta}(1+x)^{1+1/\theta}\int_{(1+x)/(1+\theta)}^Ay^{-2-1/\theta}\,\ell(y,A)\,dy,
\end{equation}
where $x\ge0$ and $A>0$, and we recall that $\ell(x,A)\triangleq\EV_\infty[\mathcal{S}_A^{r=x}]$. It is this equation that we shall now attempt solve explicitly. To that end, a natural point of departure here would be the aforementioned approximation $\ARL(\mathcal{S}_A^r)\approx A/\xi-r$, where $\xi$ is the limiting average exponential overshoot formally defined in the preceding section. It is known (see, e.g.,~\cite{Tartakovsky+Ivanova:PIT1992}) that $\xi=1/(1+\theta)\in(0,1)$ for the $\mathcal{E}(1)$-to-$\mathcal{E}(1+\theta)$ model~\eqref{eq:exp-model-def}. Hence, at least for large enough $A$'s, the solution to~\eqref{eq:exp-case-ARL-int-eqn} should behave roughly as $\ell(x,A)\approx A(1+\theta)-x$. As will be shown shortly, this is, in fact, precisely the behavior of the solution, without $A$ having to be large. However, the aforementioned fact that $\LR_n\ge1/(1+\theta)$ w.p.~1 under any measure makes things a bit complicated.
\begin{lemma}\label{lem:lwr-bnd-exp-case}
For the $\mathcal{E}(1)$-to-$\mathcal{E}(1+\theta)$ model~\eqref{eq:exp-model-def}, at each epoch $n\ge0$ and under any probability measure, the GSR statistic $R_n^{r}$ has a {\em deterministic} lower bound, i.e., $R_n^r\ge B_n^r$ w.p.~1, for each $n\ge0$ and under any probability measure, where
\begin{equation}\label{eq:lwr-bnd-exp-case}
B_n^r
\triangleq
\frac{1}{\theta}\left[1-\frac{1}{(1+\theta)^n}\right]+\frac{r}{(1+\theta)^n},\,n\ge0,
\end{equation}
and $r$ is the GSR statistic's headstart, i.e., $R_0^r=r\ge0$.
\end{lemma}
\begin{proof}
\smartqed
It is merely a matter of ``unfolding'' the recursion $R_n^r=(1+R_{n-1}^r)\LR_n$, $n\ge1$, one term at a time, and applying, at each step, the inequality $\LR_n\ge1/(1+\theta)$ valid w.p.~1 under any probability measure.
\qed
\end{proof}
At this point note that since $1+\theta>1$, the lower bound sequence $\{B_n^r\}_{n\ge0}$ given by~\eqref{eq:lwr-bnd-exp-case} is such that\begin{inparaenum}[\itshape(a)] \item for $r\le 1/\theta$, it increases monotonically with $n$, i.e., $r\equiv B_0^r\le B_1^r\le B_2^r\ldots$, when $r\le 1/\theta$, and \item $\lim_{n\to\infty}B_n^r=1/\theta$, irrespective of $R_0^r=r\ge0$\end{inparaenum}. Hence, when $A<1/\theta$, the GSR statistic, $\{R_n^r\}_{n\ge0}$, is guaranteed to either hit or exceed the level $A>0$ within at most $m$ steps, where $m\equiv m(r,A,\theta)$ is found from the inequality $B_m^r\ge A$, i.e.,
\begin{empheq}[%
    left={%
        m\equiv m(r,A,\theta)\triangleq%
    \empheqlbrace}]{align*}
&\left\lceil\left(\log\frac{1-\theta r}{1-\theta A}\right)\left/\,\log(1+\theta)\right.\right\rceil,\;\;\text{for $r<A\,(<1/\theta)$;}\\
&1,\;\;\text{for $r\ge A$,}
\end{empheq}
with $\lceil x\rceil$ denoting the usual ``ceiling'' function. Therefore, the general solution to~\eqref{eq:exp-case-ARL-int-eqn} is dependent upon whether $A<1/\theta$ or $A\ge1/\theta$. In the latter case, the (exact) solution is given by the following theorem, which is the main result of this paper.
\begin{theorem}\label{thm:exp-case-ARL-sln1}
For the $\mathcal{E}(1)$-to-$\mathcal{E}(1+\theta)$ model~\eqref{eq:exp-model-def}, if the detection threshold, $A>0$, is set so that $A\ge1/\theta$, then the ARL to false alarm of the GSR procedure is given by the formula:
\begin{equation}\label{eq:exp-case-ARL-sln1}
\ell(x,A)
=
1+(1+\theta)\left(A-\frac{1+x}{1+\theta}\right)\indicator{(1+x)/(1+\theta)\le A},
\end{equation}
and it is understood that $x\ge0$.
\end{theorem}
\begin{proof}
\smartqed
It is sufficient to insert~\eqref{eq:exp-case-ARL-sln1} into equation~\eqref{eq:exp-case-ARL-int-eqn} and directly verify that the latter does, in fact, ``check out''. The condition that $A\ge1/\theta$ ``protects'' against the situation described in Lemma~\ref{lem:lwr-bnd-exp-case} and in the discussion following it.
\qed
\end{proof}

The special case of Theorem~\ref{thm:exp-case-ARL-sln1} when $R_0^{r=x}=x=0$ (i.e., when there is no headstart) was previously established in~\cite[Proposition~1]{Kenett+Pollak:IEEE-TR1986} using the memorylessness of the exponential distribution. It is also noteworthy that formula~\eqref{eq:exp-case-ARL-sln1} as well as equation~\eqref{eq:exp-case-ARL-int-eqn} are actually valid for $x\ge-1$; the same can also be said about the general equation~\eqref{eq:ARL-int-eqn}.

We conclude this section with a brief analysis of the case when $A<1/\theta$. Recall that the integral in the right-hand side of~\eqref{eq:exp-case-ARL-int-eqn} plays no role, unless $(1+x)/(1+\theta)<A$. For this condition to hold when $A<1/\theta$, it must be the case that $(1+x)/(1+\theta)<1/\theta$, i.e., that $x<1/\theta$. Hence, if $A<1/\theta$, then $\ell(x,A)\equiv 1$ for all $x\ge1/\theta$. To obtain $\ell(x,A)$ explicitly for $x<1/\theta$, note that if $x<1/\theta$, the function $h(x)\triangleq(1+x)/(1+\theta)$, i.e., the lower limit of integration in the integral in the right-hand side of~\eqref{eq:exp-case-ARL-int-eqn}, is such that $h(x)\ge x$. As a result, the nature of the integral equation becomes such that the unknown function, $\ell(x,A)$, is dependent {\em solely} upon the values it assumes for higher $x$'s, and since $\ell(x,A)\equiv 1$ for $x\ge1/\theta$, one can iteratively work out backwards the solution for any $x\ge0$. However, this process involves formidable integrals, and only the first few steps seem to be feasible to actually carry out.

While an explicit formula for the ARL to false alarm of the GSR procedure when $A<1/\theta$ turned out to be problematic to get, from a practical standpoint it might not be worthwhile altogether, for the formula for $A\ge1/\theta$ alone, i.e., Theorem~\ref{thm:exp-case-ARL-sln1}, is sufficient. Specifically, since $\ARL(\mathcal{S}_A^r)\ge A-r$, the formula for the ARL to false alarm when $A>1/\theta$, i.e., formula~\eqref{eq:exp-case-ARL-sln1}, will never yield ARL's lower than $(1/\theta)-r$. However, the size of this ``blind spot'' is not necessarily large, unless $\theta$ is very small, which is to say that the change in the mean in the $\mathcal{E}(1)$-to-$\mathcal{E}(1+\theta)$ model~\eqref{eq:exp-model-def} is faint and not worthy of detection to begin with. As an illustration of this point, consider the original SR procedure ($r=0$) and suppose that $\theta$ is $0.01$, which, from a practical standpoint, can hardly be considered a ``change'' in the first place. Yet, since $1/\theta$ in this case is $100$, the linear formula for the ARL to false alarm will never yield a value of $100$ or less. However, this is unlikely to be of inconvenience to a practitioner, as in most applications the ARL to false alarm is set to be at least in the hundreds, and, when $\theta=0.01$, these levels of the ARL to false alarms {\em would be} obtainable through formula~\eqref{eq:exp-case-ARL-sln1}.

%+-----------------------------------------------------------------------------------------------+%
\section{Concluding Remarks}
\label{sec:conclusion}
This contribution is part of the authors' ongoing effort (manifested, e.g., in~\cite{Polunchenko+etal:SA2014,Polunchenko+etal:ASMBI2014}, and, with other collaborators, e.g., in~\cite{Tartakovsky+Polunchenko:IWAP2010,Polunchenko+Tartakovsky:AS2010,Moustakides+etal:SS2011,Tartakovsky+etal:TPA2012}) to ``pave the way'' for further research on the theory and application of the GSR procedure. To that end, case studies involving ``stress-testing'' the GSR procedure on real data are still an ``uncharted territory'' and would be of particular interest. Hopefully, the result obtained in this work, the data-analytic advantages pointed out in~\cite{Kenett+Pollak:JAP1996}, and the strong optimality properties established, e.g., in~\cite{Pollak+Tartakovsky:SS2009,Shiryaev+Zryumov:Khabanov2010,Tartakovsky+Polunchenko:IWAP2010,Polunchenko+Tartakovsky:AS2010,Tartakovsky+etal:TPA2012}, will help the GSR procedure rightly stand out as the top tool for change-point detection.

%+-----------------------------------------------------------------------------------------------+%
\begin{acknowledgement}
The authors would like to thank Prof. Sven Knoth of the Helmut Schmidt University, Hamburg, Germany, and Prof. Ansgar Steland of the RWTH Aachen University, Aachen, Germany, for the invitation to contribute this work to the 12-th German--Polish Workshop on Stochastic Models, Statistics and Their Applications. Constructive feedback provided by Dr. Ron Kenett of Israel-based KPA Ltd. (\url{www.kpa-group.com}), by Prof. William H. Woodall of Virginia Polytechnic Institute, Blacksburg, Virginia, USA, and by two anonymous referees is greatly appreciated as well.

The effort of A.S.~Polunchenko was supported, in part, by the Simons Foundation (\url{www.simonsfoundation.org}) via a Collaboration Grant in Mathematics (Award \#\,304574) and by the Research Foundation for the State University of New York at Binghamton via an Interdisciplinary Collaboration Grant (Award \#\,66761).

Last but not least, A.S.~Polunchenko is also indebted to the Office of the Dean of the Harpur College of Arts and Sciences at the State University of New York at Binghamton for the support provided through the Dean's Research Semester Award for Junior Faculty granted for the Fall semester of 2014.
\end{acknowledgement}

%+-----------------------------------------------------------------------------------------------+%
%%%%%%%%%%%%%%%%%%%%%%%% referenc.tex %%%%%%%%%%%%%%%%%%%%%%%%%%%%%%
% sample references
% %
% Use this file as a template for your own input.
%
%%%%%%%%%%%%%%%%%%%%%%%% Springer-Verlag %%%%%%%%%%%%%%%%%%%%%%%%%%
%
%
%\biblstarthook{References should be given by author/year.\footnote{Make sure that all references from the list are cited in the text.} The reference list should ideally be \textit{sorted} in alphabetical order -- even if reference numbers are used for the citation in the text. If there are several works by the same author, the following order should be used:
%\begin{enumerate}
%\item all works by the author alone, ordered chronologically by year of publication
%\item all works by the author with a coauthor, ordered alphabetically by coauthor
%\item all works by the author with several coauthors, ordered chronologically by year of publication.
%\end{enumerate}
%The \textit{styling} of references\footnote{Always use the standard abbreviation of a journal's name according to the ISSN \textit{List of Title Word Abbreviations}, see \url{http://www.issn.org/en/node/344}} should follow the examples given for a contribution \cite{basic-contrib}, an online resource \cite{basic-online}, journal articles \cite{basic-journal}, DOI references \cite{basic-DOI} and monographs \cite{basic-mono}.
%}

\end{document}